\documentclass{llncs}[11pt]
\usepackage[left=1.75in, right=1.75in]{geometry}
\usepackage{amsmath}
\usepackage{amssymb}

\begin{document}
\frontmatter          
\pagestyle{headings}  

\mainmatter   
%
%
\title{Parameterized and Approximation Algorithms for Boxicity}
%

%
\titlerunning{Approximation Algorithms for Boxicity}  
\author{Abhijin Adiga \and Jasine Babu \and L. Sunil Chandran}
\authorrunning{Adiga et.al} 
\institute{Department of Computer Science and Automation,\\Indian Institute of Science, Bangalore 560012, India.\\
\email{abhijin@gmail.com, \{jasine, sunil\}@csa.iisc.ernet.in}}
\maketitle  
%
%
 


\begin{abstract}
Boxicity of a graph $G(V,$ $E)$, denoted by $box(G)$, is the minimum  integer $k$ such that $G$ can be represented as the intersection graph of axis parallel boxes in $\mathbb{R}^k$. The problem of computing boxicity is inapproximable even for graph classes like bipartite, co-bipartite and split graphs within $O(n^{1 - \epsilon})$-factor, for any $\epsilon >0$ in polynomial time unless $NP=ZPP$. We give FPT approximation algorithms for computing the boxicity of graphs, where the parameter used is the vertex or edge edit distance of the given graph from families of graphs of bounded boxicity. This can be seen as a generalization of the parameterizations discussed in \cite{Adiga2}. 

Extending the same idea in one of our algorithms, we also get an $O\left(\frac{n\sqrt{\log \log n}}{\sqrt{\log n}}\right)$ factor approximation algorithm for computing boxicity and an $O\left(\frac{n {(\log \log n)}^{\frac{3}{2}}}{\sqrt{\log n}}\right)$ factor approximation algorithm for computing the cubicity.  These seem to be the first $o(n)$ factor approximation algorithms known for both boxicity and cubicity. As a consequence of this result, 
a $o(n)$ factor approximation algorithm for computing the partial order dimension of finite posets and a $o(n)$ factor approximation algorithm for computing the threshold dimension of split graphs would follow.
\end{abstract} 
 \section{Introduction}
Let $G(V$, $E)$ be a graph. If  $I_1$, $I_2$, $\cdots$, $I_k$ are (unit) interval graphs on the vertex set $V$ such that $E(G)=E(I_1) \cap E(I_2) \cap \cdots \cap E(I_k)$, then $\{I_1$, $I_2$, $\cdots$, $I_k\}$ is called a box (cube) representation of $G$ of dimension $k$. Boxicity (cubicity) of an incomplete graph $G$, $box(G)$ (respectively $cub(G)$), is defined as the minimum integer $k$ such that $G$ has a box (cube) representation of dimension $k$. For a complete graph, it is defined to be zero. Equivalently, boxicity (cubicity) is the minimum integer $k$ such that $G$ can be represented as the intersection graph of axis parallel boxes (cubes) in $\mathbb{R}^k$. Boxicity was introduced by Roberts \cite{Rob1} in 1969 for modeling problems in social sciences and ecology. Box representations of low dimension are memory efficient for representing dense graphs. If a graph $G$ on $n$ vertices has a box representation of dimension $k$, it can be represented using $O(nk)$ space, whereas an adjacency list 
representation will need $O(m)$ space which is $O(n^2)$ for dense graphs. Some well known NP-hard problems like the max-clique problem are polynomial time solvable, if low dimensional box representations are known \cite{Rosgen}.

Boxicity is combinatorially well studied and its bounds in terms of parameters like maximum degree \cite{AdigaCOCOON,Esperet09} and tree-width \cite{Chandran2007} are known. For any graph $G$ on $n$ vertices, $box(G) \le \left \lfloor \frac{n}{2}\right \rfloor$ and $cub(G) \le \left \lfloor \frac{2n}{3}\right \rfloor$. It was shown by Scheinerman \cite{Sch1} in 1984 that the boxicity of outer planar graphs is at most two. In 1986, Thomassen \cite{Thom1} proved that the boxicity of planar graphs is at most 3. Boxicity is also closely related to other dimensional parameters of graphs like partial order dimension and threshold dimension \cite{AdigaCOCOON,Abh1,Yan1}.

However, from the computational point of view, boxicity is a notoriously hard problem. In 1981, Cozzens\cite{Coz1} proved that computing boxicity is NP-Hard. Later Yannakakis \cite{Yan1} and Kratochvil\cite{Krat1}, respectively, proved that determining whether boxicity of a graph is at most three and two are NP-Complete. Adiga et.al \cite{Abh1} proved that no polynomial time algorithm for approximating boxicity of bipartite graphs with approximation factor within 
$O(n^{0.5 - \epsilon})$ for any $\epsilon > 0$ is possible unless $NP=ZPP$. Recently, Chalermsook et al. \cite{Chalermsook2013} improved the hardness of approximation from $O(n^{0.5 - \epsilon})$ factor to $O(n^{1 - \epsilon})$ factor. Same non-approximability holds in the case of split graphs and co-bipartite graphs too. 

Since boxicity is even hard to approximate, one would like to look at parameterized\footnote[4]{For an introduction to parameterized complexity, please refer to \cite{Nie1}.} versions of the problem. The standard parameterization using boxicity as parameter is meaningless since deciding whether boxicity is at most $k$ is NP-Hard even for $k=2$. Parameterizations with vertex cover number (MVC) and minimum feedback vertex set size (FVS) as parameters were studied in \cite{Adiga2}. With vertex cover number as parameter $k$, they gave an algorithm to compute boxicity exactly, that runs in $2^{O(2^k k^2)}n$ time and another algorithm to get an additive one approximation for boxicity that runs in $2^{O(k^2 \log k )} n$ time, where $n$ is the number of vertices in the graph. Using FVS as parameter $k$, they gave a $2+\frac{2}{box(G)}$ factor approximation algorithm to compute boxicity that runs in $2^{O(2^k k^2)}n^{O(1)}$ time. 

The notion of edit distance refers in general to the smallest number of some well defined modifications to be applied to the input graph so that the resultant graph possesses some desired properties. Edit distance from graph classes is a well studied problem in parameterized complexity \cite{Cai96,Grohe2004,Marx2007,Vil1}. Note that, many well known classical problems such as interval completion problem \cite{Vil1}, planar vertex deletion problem \cite{Marx2007} etc. can be seen as special instances of edit distance problems. In \cite{Vil1}, it is shown that the problem of interval completion of graphs is in FPT, with interval completion number (the minimum number of edges to be added to a graph to convert into an interval graph) as the parameter. In \cite{Grohe2004,Marx2007} problems of editing graphs to make them planar is considered. 

In \cite{Cai03}, Cai introduced a framework for parameterizing problems with edit distance as parameter. For a family $\mathcal{F}$ of graphs, and $k\ge 0$ an integer, he used $\mathcal{F} + ke$, $\mathcal{F} - ke$ respectively to denote the families of graphs that can be obtained from a graph in $\mathcal{F}$ by adding and deleting at most $k$ edges, and $\mathcal{F} + kv$ to denote the family of graphs that can be converted to a graph in $\mathcal{F}$ by deleting at most $k$ vertices. A subset $S\subseteq V$ such that $|S|\le k$ is called a\textbf{ modulator} for an $\mathcal{F} + kv$ graph $G(V, E)$ if $G\setminus S \in \mathcal{F}$. Similarly, a set $E_k$ of pairs of vertices such that $|E_k|\le k$ is called a modulator for an $\mathcal{F} - ke$ graph $G(V, E)$ if $G'(V, E \cup E_k) \in \mathcal{F}$. In a similar way, modulators for graphs in $\mathcal{F} + ke$ and $\mathcal{F} + k_1e -k_2e$ also can be defined.

In \cite{Cai03}, Cai considered the parameterized complexity of vertex coloring problem on $\mathcal{F} -ke$, $\mathcal{F} + ke$ and $\mathcal{F} + kv$ for various families $\mathcal{F}$ of graphs, with $k$ as the parameter. This was further studied in \cite{Marx06}. In the same framework, we consider the parameterized complexity of computing boxicity of $\mathcal{F} + k_1e -k_2e$ and $\mathcal{F} + kv$ graphs for families $\mathcal{F}$ of bounded boxicity graphs, using $k_1+k_2$ and $k$ as parameters. We show that various other parameters relevant in the context of boxicity problem (such as interval completion number, MVC, FVS, crossing number etc.) can be seen as special instances of our edit distance parameters. (Note that the special cases of MVC and FVS were considered in \cite{Adiga2}. We provide an improved algorithm for the parameter FVS and we show that two of our parameters are more general than MVC.)
\subsection*{Our results}
Our main results are the following three theorems and their corollaries described below.
\begin{theorem}\label{thved}
Let $\mathcal{F}$ be a family of graphs such that $\forall G' \in \mathcal{F}$, $box(G') \le b$. Let $T(n)$ denote the time required to compute a $b$-dimensional box representation of a graph belonging to $\mathcal{F}$ on $n$ vertices. Let $G$ be an $\mathcal{F}+kv$ graph on $n$ vertices. Given a modulator of $G$, a box representation $\mathcal{B}$ of $G$, such that $|\mathcal{B}| \le box(G)\left(2+\frac{b}{box(G)}\right)$, can be computed in time $T(n-k)+n^2 2^{{O(k^2 \log k)}}$.
\end{theorem}
\begin{theorem}\label{theed}
Let $\mathcal{F}$ be a family of graphs such that $\forall G' \in \mathcal{F}$, $box(G') \le b$. Let $T(n)$ denote the time required to compute a $b$-dimensional box representation of a graph belonging to $\mathcal{F}$ on $n$ vertices. Let $G$ be an $\mathcal{F}+k_1e-k_2e$ graph on $n$ vertices and let $k=k_1+k_2$. Given a modulator of $G$, a box representation $\mathcal{B}$ of $G$, such that $|\mathcal{B}| \le box(G) + 2b$, can be computed in time $T(n)+O(n^2)+ 2^{{O(k^2 \log k)}}$. 
\end{theorem}
 Slightly modifying the idea in the proof of Theorem \ref{thved}, we also prove the following:
\begin{theorem}\label{thgen}
 Let $G(V,E)$ be a graph on $n$ vertices. Then a box representation $\mathcal{B}$ of $G$, such that $|\mathcal{B}| \le t \cdot box(G)$, where $t$ is $O\left(\frac{n\sqrt{\log \log n}}{\sqrt{\log n}}\right)$, can be computed in polynomial time. Further, a cube representation $\mathcal{C}$ of $G$, such that $|\mathcal{C}| \le t' \cdot cub(G)$, where $t'$ is $O\left(\frac{n {(\log \log n)}^{\frac{3}{2}}}{\sqrt{\log n}}\right)$, can also be computed in polynomial time.
\end{theorem}
We give proofs of these theorems in Sections \ref{fpt2}, \ref{fpt1} and \ref{approx}. The approximation factors obtained in Theorem \ref{thgen} is not very impressive. However, boxicity and cubicity problems are known to be inapproximable within $O(n^{1 - \epsilon})$-factor, for any $\epsilon >0$ unless $NP=ZPP$ and to our knowledge, no approximation algorithms for computing boxicity and cubicity of general graphs within $o(n)$ factor were known till now. We also show that as a consequence of the $o(n)$ factor approximation algorithm for boxicity, a $o(n)$ factor approximation algorithm for computing the partial order dimension of finite posets and a $o(n)$ factor approximation algorithm for computing the threshold dimension of split graphs would follow. These problems are also known to be hard to approximate within $O(n^{1 - \epsilon})$-factor, for any $\epsilon >0$ unless $NP=ZPP$ \cite{Chalermsook2013}.\\\\
\textbf{Remarks:} 
\begin{itemize}
    \item[(1)] Though in Theorem \ref{thved} and Theorem \ref{theed} we assumed that a modulator of $G$ for $\mathcal{F}$ is given, in several important special cases, the modulator for $\mathcal{F}$ can be computed from $G$ in FPT time. This will be clear from the corollaries that follow. 
    \item[(2)] For many graph classes $\mathcal{F}$ which we consider here (like interval graphs or planar graphs), we can compute a low dimensional box representation of graphs belonging to $\mathcal{F}$ in polynomial time. Note that, in such cases, if the edit distance $k$ is below $\frac{\sqrt{\log n}}{\sqrt{\log \log n}}$, the algorithms mentioned in Theorem \ref{thved} and Theorem \ref{theed} run in time polynomial in $n$.
   \item[(3)] Note that, if $G \in \mathcal{F}+k_1e-k_2e$ with $k=k_1+k_2$, then $G \in \mathcal{F}+kv$. Therefore, the algorithm of Theorem \ref{thved} is applicable for graphs in $\mathcal{F}+k_1e-k_2e$ as well. But we give a more specialized algorithm for $\mathcal{F}+k_1e-k_2e$ graphs, which gives us an improved result as mentioned in Theorem \ref{theed}.
  \end{itemize}
\textbf{Corollaries:}
Our parameterized algorithms for boxicity generalizes parameterizations using many other useful parameters, as explained below. Some of these parameterizations are new and the others are improvements / generalizations of the results in \cite{Adiga2}.
\begin{itemize}
\item[(1)] 
\textbf{Computing boxicity with interval completion number as the parameter:} If the interval completion number of a graph $G(V, E)$ is at most $k$, we can use FPT algorithm for interval completion \cite{Vil1} to compute $E_k$ such that $|E_k| \le k$ and $G'(V, E \cup E_k)$ is an interval graph. Thus, with modulator $E_k$, $G \in \mathcal{F}-ke$, where $\mathcal{F}$ is the class of interval graphs. Since a box representation of one dimension can be computed in polynomial time for any interval graph, combining with our algorithm of Theorem \ref{theed}, we get an FPT algorithm that achieves an additive 2 factor approximation for $box(G)$, with interval completion number $k$ as parameter running in time $2^{O(k^2 \log k )}n^{O(1)}$. 
\item[(2)] 
\textbf{Computing boxicity with FVS as the parameter:} If $FVS(G) \le k$, using existing FPT algorithms \cite{cao10}, we can compute a minimum feedback vertex set $S$ of $G(V, E)$ such that $G'=G(V\setminus S)$ is a forest. Thus, with modulator $S$, $G \in \mathcal{F}+kv$, where $\mathcal{F}$ is the family of graphs which are forests. Since a box representation of dimension two can be computed in polynomial time for any forest, using our  algorithm of Theorem \ref{thved}, we get a $2+\frac{2}{box(G)}$ factor approximation for boxicity with FVS as parameter $k$, running in time $2^{O(k^2 \log k )}n^{O(1)}$. Note that, for the boxicity problem parameterized by FVS, the algorithm in \cite{Adiga2} gave the same approximation factor with its running time $2^{O(2^k k^2)}n^{O(1)}$. Our algorithm improves the running time in \cite{Adiga2}.
\item[(3)]
\textbf{Computing boxicity with proper interval vertex deletion number (PIVD) as the parameter:} The minimum number of vertices to be deleted from $G(V, E)$, so that the resultant graph is a proper interval graph, is called the proper interval vertex deletion number of $G$. If $PIVD(G)$ is at most $k$, we can use the FPT algorithm running in $O(6^k  kn^6)$ time for proper interval vertex deletion \cite{vil2} to compute a $S\subseteq V$ with $|S|\le k$ such that $G \setminus S$ is a proper interval graph. Thus, with modulator $S$, $G \in \mathcal{F}+kv$, where $\mathcal{F}$ is the family of all proper interval graphs. Since a box representation of one dimension can be computed in polynomial time for any proper interval graph, using our  algorithm of Theorem \ref{thved}, we get a $2+\frac{1}{box(G)}$ factor approximation for boxicity with PIVD as parameter $k$, running in time $2^{O(k^2 \log k )}n^{O(1)}$. We show that the parameter PIVD is more general than the MVC parameter considered in \cite{Adiga2}. 
\item[(4)]
\textbf{Computing boxicity with MVC as the parameter:}
$MVC(G)$ can be seen as the minimum number of vertices to be deleted from $G$ so that the resultant graph $G' \in \mathcal{F}'$ where $\mathcal{F}'$ is the family of graphs without any edges. Since $\mathcal{F}' \subseteq \mathcal{F}$, the family of proper interval graphs, it is easy to see that $PIVD(G)  \le MVC(G)$. Therefore, for computing boxicity with MVC as parameter, we can use the algorithm above for computing boxicity with PIVD as parameter. This gives us an algorithm for boxicity with MVC as parameter $k$, that achieves a $2+\frac{1}{box(G)}$ factor approximation with the same running time as that of the (better) additive one factor approximation algorithm for boxicity with parameter MVC, described in \cite{Adiga2}. However, as explained, the parameter PIVD above is more general and could be much smaller than MVC.  
\item[(5)] 
\textbf{Computing boxicity with planar vertex deletion number as the parameter:} The minimum integer $k$ such that $G(V, E)$ is a Planar$+kv$ graph, is called the planar vertex deletion number of $G$. If $G \in$ Planar$+kv$, we can use the FPT algorithm for planar deletion \cite{Marx2007} to compute a $S\subseteq V$ with $|S|\le k$ such that $G \setminus S$ is planar. Thus, with modulator $S$, $G \in \mathcal{F}+kv$, where $\mathcal{F}$ is the family of planar graphs. Since planar graphs have 3 dimensional box representations computable in polynomial time \cite{Thom1}, using our  algorithm of Theorem \ref{thved}, we get an FPT algorithm for boxicity, giving a $2+\frac{3}{box(G)}$ factor approximation for boxicity of graphs that can be made planar by deleting at most $k$ vertices, using planar vertex deletion number as parameter. It may be noted that this parameter is also smaller than MVC.
\item[(6)] 
\textbf{Computing boxicity with crossing number as the parameter:} If crossing number of a graph is at most $k$, we can combine the FPT algorithm for crossing number \cite{Grohe2004} to compute $E_k \subseteq E$ such that $|E_k| \le k$ and $G'(V, E \setminus E_k)$ is a planar graph. Thus, with modulator $E_k$, $G \in \mathcal{F}+ke$, where $\mathcal{F}$ is the class of planar graphs. Since planar graphs have 3 dimensional box representations computable in polynomial time, using our algorithm of Theorem \ref{theed}, we get an FPT algorithm that gives an additive 6-factor approximation for $box(G)$ with crossing number as parameter.
\item[(7)]
\textbf{Computing boxicity with planar edge deletion number as parameter:} Planar edge deletion number of a graph $G(V, E)$ is the minimum number of edges to be deleted from $G$ so that the resultant graph is planar. In \cite{Grohe2004}, an FPT algorithm for computing planar edge deletion number is also described. Using the same ideas as in the case of crossing number, we get an FPT algorithm that gives an additive 6-factor approximation for $box(G)$ with planar edge deletion number as parameter. Since planar edge deletion number$(G)$ $\le$ crossing number$(G)$, this parameter is more general than crossing number.
\item[(8)]
\textbf{Computing boxicity with proper interval edge deletion number (PIED) as the parameter:} The minimum number of edges to be deleted from $G(V, E)$, so that the resultant graph is a proper interval graph, is called the proper interval edge deletion number of $G$. If $PIED(G)$ is at most $k$, we can use the FPT algorithm running in $O(9^k n^{O(1)})$ time for proper interval edge deletion \cite{vil2} to compute a $E_k \subseteq E$ with $|E_k|\le k$ such that $G'(V, E \setminus E_k)$ is a proper interval graph. Thus, with modulator $S$, $G \in \mathcal{F}+ke$, where $\mathcal{F}$ is the family of all proper interval graphs. Since a box representation of one dimension can be computed in polynomial time for any interval graph, combining with our algorithm of Theorem \ref{theed}, we get an FPT algorithm that achieves an additive 2 factor approximation for $box(G)$, with PIED as parameter $k$, running in time $2^{O(k^2 \log k )}n^{O(1)}$. 
\end{itemize}
\textbf{FPT algorithm for cubicity: }
Computing cubicity is also hard to approximate \cite{Chalermsook2013} within $O(n^{1 - \epsilon})$ factor for any $\epsilon > 0$ unless $NP=ZPP$. It is natural to ask, like in the case of boxicity, whether FPT algorithms are possible for cubicity as well, with various edit distance parameters. Unfortunately, our algorithms of Theorems \ref{thved} and \ref{theed} heavily depend on the fact that intervals can be of different lengths. Since for cube representations all intervals are required to be of unit length, there is no direct way to extend our algorithms for cubicity. This, we leave as an open problem. However, in the special case of parameter $MVC(G)$, which is a relatively simple edit distance parameter, in Section \ref{fpt3}, we give a 2-factor approximation algorithm which runs in time $2^{O(2^k k^2)} n^{O(1)}$, where $k=MVC(G)$. This algorithm can be modified to get an approximation factor $(1+\epsilon)$, for any $\epsilon >0$, by allowing a larger running time of $2^{O\left(\frac{k^3 2^{\frac{4k}{\epsilon}}}{
\epsilon}\right)} n^{O(1)}$. 
\section{Prerequisites} \label{prereq} 
In this section, we give some basic facts necessary for the later part of the paper. 
\begin{lemma}[Roberts \cite{Rob1}] \label{lmrob}
Let $G(V,$ $E)$ be any graph. For any $x\in V$, $box (G) \le 1 + box (G \setminus \{x\})$.
\end{lemma}
The proof of the following lemmas are easy.
\begin{lemma}\label{lmeasy}
Let $G(V,E)$ be a graph. Let $S\subseteq V$ be such that $\forall v \in V \setminus S$ and $u \in V$, $(u,v) \in E$. If a box representation $\mathcal{B}_S$ of $G[S]$ is known, then, in $O(n^2)$ time we can construct a box representation $\mathcal{B}$ of $G$ of dimension $|\mathcal{B}_S|$. In particular, $box(G) = box(G[S])$.
\end{lemma}
Please see Appendix for the proof.
\begin{lemma}\label{lm4}
  Let $G(V, E)$ be a graph and let $A \subseteq V$. Let $G_1(V, E_1)$ be a super graph of $G$ with $E_1 = E \cup \{(x,y)|x,y \in A\}$. If a box representation $\mathcal{B}$ of $G$ is known, then in $O(n^2)$ time we can construct a box representation $\mathcal{B}$$_1$of $G_1$ of dimension $2 \cdot |\mathcal{B}|$. In particular, $box(G_1) \le 2.box(G)$.\end{lemma}
Please see Appendix for the proof.
\begin{lemma}\label{count}
 Let $G(V, E)$ be a graph on $n$ vertices of boxicity (cubicity) $b$. Then an optimum box (cube) representation of $G$ can be computed in $2^{O(nb \log n)}$ time.
\end{lemma}
Please see Appendix for the proof.
\\\\
For a vertex $v \in V$ of a graph $G$, we use $N_G(v)$ to denote the set of neighbors of $v$ in $G$. Let $I$ be an interval representation of an interval graph $G(V, E)$. We use $l_v(I)$ and $r_v(I)$ respectively to denote the left and right end points of the interval corresponding to $v\in V$ in $I$. The interval is denoted as $\bigl[l_v(I), r_v(I)\bigr]$. Without loss of generality, we can assume that all the $2|V|$ interval end points are distinct points in $\mathbb{R}$. Unless specified otherwise, we make this as a default assumption. If $S \subseteq V$ induces a clique in $G$, then, it is easy to see that the intersection of all the intervals in $I$ corresponding to vertices of $S$ is nonempty. This property is referred to as \textit{Helly property of intervals} and we refer to this common region of intervals as the Helly region of the clique $S$.
\begin{definition}
 Let $G(V, E)$ be a graph in which $S \subseteq V$ induces a clique in $G$. Let $H$ be an interval supergraph of $G$. Let $p$ be a point on the Real line. If $H$ has an interval representation $I$ satisfying the following conditions:
\begin{itemize}
 \item[(1)] $p$ belongs to the Helly region of $S$ in $I$.
 \item[(2)] For each $v \in S$,\\ $l_v(I)=\min \left( p,\displaystyle\min_{u \in N_G(v) \cap (V \setminus S)} {r_u(I)} \right)$ and\\$r_v(I)=\max \left( p,\displaystyle \max_{u \in N_G(v)\cap (V \setminus S)} {l_u(I)}\right)$
\end{itemize}
then we call $I$ a nice interval representation of $H$ with respect to $S$ and $p$. If $H$ has a nice interval representation with respect to clique $S$ and some point $p$, then $H$ is called  a nice interval supergraph of $G$ with respect to clique $S$. 
\end{definition}
\begin{lemma}\label{thmnice}
 Let $G(V, E)$ be a graph. If $A \subseteq V$ with $|A| \le k$ and $G[V \setminus A]$ a clique on $V \setminus A$, then 
 \begin{itemize}
  \item[(a)] There are at most $2^{O(k \log k)}$ nice interval supergraphs of $G$ with respect to clique $V \setminus A$. These can be enumerated in $n2^{O(k \log k)}$ time.
  \item[(b)] If $G$ has a box representation $\mathcal{B}$ of dimension $b$, then it has a box representation $\mathcal{B'}$ of the same dimension, in which $\forall I \in \mathcal{B'}$, $I$ is a nice interval supergraph of $G$ with respect to clique $V \setminus A$.
\end{itemize}
\end{lemma}
\begin{proof}
  \begin{itemize}
  \item[(a)] Let $H$ be any nice interval super graph of $G$ with respect to $V \setminus A$ and $I$ be a nice interval representation of $H$ with respect to $V \setminus A$ and a point $p$. Let $S$ be the set of end points (both left and right) of the intervals corresponding to vertices of $A$ in $H$. Clearly $|S|=2|A|\le 2k$. The order of end points of vertices of $A$ in $I$ from left to right corresponds to a permutation of elements of $S$ and therefore, there are at most $2k!$ possibilities for this ordering. Moreover, note that the points of $S$ divide the Real line into $|S|+1$ regions and that $p$ can belong to any of these regions. From the definition of nice interval representation, it is clear that, once the point $p$ and the end points of vertices of $A$ are fixed, the end points of vertices in $V\setminus A$ get automatically decided. 

Thus, to enumerate every nice interval supergraph $H$ of $G$ with respect to clique $V \setminus A$, it is enough to enumerate all the $(2k)!=2^{O(k \log k)}$ permutations of elements of $S$ and consider $|S|+1 \le 2k+1$ possible placements of $p$ in each of them. Some of these orderings may not produce an interval super graph of $G$ though. In $O(k^2)$ time, we can check whether the resultant graph is an interval supergraph of $G$ and output the interval representation in $O(n)$ time. The number of supergraphs enumerated is only $(2k+1)2^{O(k \log k)}= 2^{O(k \log k)}$. 
\item[(b)] Let $\mathcal{B} = \{I_1$, $I_2$, $\cdots$, $I_b\}$ be a box representation of $G$ and for $1 \le i \le b$, let $p_i \in \mathbb{R}$ be a point belonging to the Helly region corresponding to $V \setminus A$ in $I_i$. For $1 \le i \le b$, let $I'_i$ be the interval graph defined by the interval assignment
$$
\left[ l_v(I'_i),r_v(I'_i)\right] = 
\begin{cases}  [l_v(I_i),r_v(I_i)] & \text{if $v\in A$,}
\\
[l'_v(i),r'_v(i)] &\text{if $v \in V\setminus A$.}
\end{cases}
$$
where $l'_v(i)=\min \left( p_i,\displaystyle\min_{u \in N_G(v) \cap A} {r_u(I_i)} \right)$ and\\$r'_v(i)=\max \left( p_i,\displaystyle \max_{u \in N_G(v)\cap A} {l_u(I_i)}\right)$\\\\ 
\begin{claim}\label{claimNice}
 $\mathcal{B'}= \{I'_1$, $I'_2$, $\cdots$, $I'_b\}$ is a box representation of $G$ such that $\forall I'_i \in \mathcal{B'}$, $I'_i$ is a nice interval supergraph of $G$ with respect to clique $V \setminus A$.
\end{claim}
\begin{proof}
 Consider any $I'_i \in \mathcal{B'}$. For $u, v \in A$, intervals corresponding to $u$ and $v$ are the same in both $I_i$ and $I'_i$.  If $(u, v) \in E(G)$, with $u, v \in A$, then the intervals corresponding to $u$ and $v$ intersect in $I'_i$ because they were intersecting in $I_i$. For any $(u, v) \in E(G)$, with $u \in A$ and $v \in V \setminus A$, the interval of $v$ intersects the interval of $u$ in $I'_i$, by the definition of $[l'_v(i),r'_v(i)]$. Vertices of $V \setminus A$ share the common point $p_i$. Thus, $I'_i$ is an interval supergraph of $G$. It is easy to see that $I'_i$ is a nice interval supergraph of $G$ with respect to clique $V \setminus A$ and point $p_i$.

Since $\mathcal{B}$ is a valid box representation of $G$, for each $(u, v) \notin E(G)$, $\exists I_i \in \mathcal{B}$ such that $(u, v) \notin E(I_i)$. Observe that for any vertex $v \in V$, interval of $v$ in $I_i$ contains the interval of $v$ in $I'_i$. Therefore, if $(u, v) \notin E(I_i)$, then $(u, v) \notin E(I'_i)$ too. Thus, $\mathcal{B'}$ is also a valid box representation of $G$. 
\end{proof}
\end{itemize}
\qed
\end{proof}
The following theorem is the key ingredient in our parameterized algorithm with vertex edit distance parameter and also our general approximation algorithm for boxicity.
\begin{theorem}\label{th4}
  Let $G(V, E)$ be a graph. If $A \subseteq V$ with $|A| \le k$ and $G[V \setminus A]$ a clique on $V \setminus A$, then
 \begin{itemize}
  \item[(a)] $box(G) \le k$. 
  \item[(b)] An optimal box representation of $G$ can be found in time $n^2 2^{{O(k^2 \log k)}}$, where $n=|V|$.
 \end{itemize}
\end{theorem}
\begin{proof}
 \begin{itemize}
  \item[(a)] It is easy to infer from Lemma \ref{lmrob} that $box(G)\le box ( G \setminus A)+ |A|$ $=$ $k$ since $G\setminus A$ is a clique. 
  \item[(b)] From part (b) of Lemma \ref{thmnice}, it is easy to see that if $box(G)=b\le k$, then there exists a box representation $\mathcal{B'}= \{I'_1$, $I'_2$, $\cdots$, $I'_b\}$ in which each $I'_i$ is a nice interval super graph of $G$ with respect to clique $V \setminus A$. We call such a representation a nice box representation of $G$. To construct a nice box representation of $G$ with respect to clique $V \setminus A$ of dimension $b$, we do the following: We choose $b$ of the $2^{O(k \log k)}$ supergraphs guaranteed by part (a) of Lemma \ref{thmnice} and check if this gives a valid box representation of $G$. All possible nice box representations of dimension $b$ can be computed and validated in $n^2 2^{{O(k.b \log k)}}$ time. We might have to repeat this process  for $1 \le b \le box(G)$ in that order, to obtain a minimum box representation. Hence the total time required to compute an optimum box representation of $G$ is $n^2 2^{{O(k^2 \log k)}}$.
\end{itemize}
\qed
\end{proof}
\section{FPT Algorithm for Computing the Boxicity of $\mathcal{F}+kv$ Graphs} \label{fpt2} 
In this section, we give a proof of Theorem \ref{thved}. 
Let $G(V, E)$ be a $\mathcal{F}+kv$ graph with a modulator $S_k$ on $k$ vertices such that $G'= G \setminus S_k \in \mathcal{F}$. Let $H_1(V, E_1)$ be the graph obtained by defining $E_1=E \cup \{(x,y)|x,y \in V \setminus S_k\}$. Since $|S_k| \le k$, using Theorem \ref{th4}, we can get an optimal box representation of $H_1$ in $n^2 2^{{O(k^2 \log k)}}$ time. Let $\mathcal{B}_1 = \{I_1$, $I_2$, $\cdots$, $I_p\}$ be the resultant box representation of $H_1$. By Lemma \ref{lm4}, $p$ is at most $2\cdot box(G)$. 

Let $H_2(V, E_2)$ be the graph obtained by defining $E_2=E \cup \{(x,y)|x \in S_k$ and $y \in V\}$. Let $\mathcal{B'}$ $=\{J_1$, $J_2$, $\cdots$, $J_b\}$ be a box representation of $G'$ (computed in time $T(n-k)$). $\mathcal{B'}$ is a box representation of $H_2[V \setminus S_k]$, because $H_2[V \setminus S_k]=G'$. Since in $H_2$, vertices in $S$ are adjacent to every other vertex, by Lemma \ref{lmeasy}, $box(H_2)=box(H_2[V \setminus S_k]$ and a box representation $\mathcal{B}_{2}$ $=\{L_1$, $L_2$, $\cdots$, $L_b\}$ of $H_2$ can be produced in $O(n^2)$ time.

Since $G=H_1 \cap H_2$, $\mathcal{B}$ $=$ $\mathcal{B}_{1}$ $\cup$ $\mathcal{B}_{2}$ as a valid box representation of $G$, of dimension $box(G)\left(2+\frac{box(G')}{box(G)}\right)$. All computations were done in $T(n-k)+n^2 2^{{O(k^2 \log k)}}$ time.
\section{FPT Algorithm to Compute Boxicity of $\mathcal{F}+k_1e-k_2e$ Graphs} \label{fpt1}
In this section, we give a proof of Theorem \ref{theed}. Let $G(V, E)$ be a $\mathcal{F}+k_1e-k_2e$ graph on $n$ vertices, where $k_1+k_2=k$. Let $E_{k_1} \cup E_{k_2}$ be a modulator of $G$ such that $|E_{k_1}|=k_1$, $|E_{k_2}|=k_2$ and $G'(V, \left(E \cup E_{k_2}\right) \setminus E_{k_1}) \in \mathcal{F}$.
Let $S \subseteq V(G)$ be the set of end points of the edges in $E_{k_1}\cup E_{k_2}$. 
 Clearly, $|S| \le 2k$ and $box(G) \le k$. Using the construction in Lemma \ref{count}, an optimal box representation $\mathcal{B}_S = \{I_1$, $I_2$, $\cdots$, $I_p\}$ of $G[S]$ can be computed in $2^{O(k^2 \log k)}$ time. 

Let $b$ be the boxicity of $G'$ and let $\mathcal{B'}$ $=\{J_1$, $J_2$, $\cdots$, $J_b\}$ be an optimal box representation of $G'$ (computed in time $T(n)$). We will produce a near optimal box representation of $G$ using box representations $\mathcal{B}_S$ and $\mathcal{B'}$, thus giving a constructive proof. 

Let $H_1(V, E_1)$ be the graph obtained by setting $E_1=E' \cup \{(x,y)|x, y \in S\}$. From the box representation $\mathcal{B'}$ of $G'$, in $O(n^2)$ time we can construct (by Lemma \ref{lm4}) a box representation $\mathcal{B}_1$ $=\{J_{11}$, $J_{12}$, $J_{21}$, $J_{22}$, $\cdots$, $J_{b1}$, $J_{b2}\}$ of $H_1$ with dimension $2 \cdot box(G')$. 

Let $H_2(V, E_2)$ be the graph obtained by setting $E_1=E \cup \{(x,y)|x \in V \setminus S, y \in V \}$. Observe that $H_2(S)=G[S]$. By lemma \ref{lmeasy}, $box(H_2) \le box(G[S])$ and a box representation $\mathcal{B}_2$ of $H_2$ of dimension $box(G[S])$ can be computed from box representation $\mathcal{B}_S$ of $G[S]$ in $O(n^2)$ time.

We describe how to compute a box representation of $G$ from box representations $\mathcal{B}_S$ and $\mathcal{B'}$.
Let $H_1$ and $H_2$ be as defined above. We constructed box representation $\mathcal{B}_1$ of $H_1$ of dimension $2 \cdot box(G')$ from $\mathcal{B'}$ in polynomial time. Box representation $\mathcal{B}_S$ was obtained in $2^{{O(k^2 \log k)}}$ time. Box representation $\mathcal{B}_2$ of $H_2$ of dimension $box(G[S])$ was constructed using $\mathcal{B}_S$ in $O(n^2)$ time. It is easy to see that $G = H_1 \cap H_2$ and hence $\mathcal{B}_G$ $=$ $\mathcal{B}_1$ $\cup$ $\mathcal{B}_2$ is a valid box representation of $G$ of dimension at most $2 \cdot box(G')$ $+$ $box(G[S]) \le 2 \cdot box(G') + box(G)$, since $G[S]$ is an induced subgraph of $G$. 
\section{An Approximation Algorithm for Boxicity of Graphs} \label{approx}
In this section, we give a proof of Theorem \ref{thgen} and use it to derive sublinear approximation algorithms for some other dimensional parameters closely related to boxicity.

Let $G(V, E)$ be the given graph with $|V|=n$. Let $k = \frac{\sqrt{\log n}}{\sqrt{\log \log n}}$ and $t = \lceil \frac{n}{k} \rceil$. The algorithm proceeds by defining $t$ super graphs of $G$ and computing their box representations. Let the vertex set $V$ be partitioned arbitrarily into $t$ sets $V_1, V_2, \cdots, V_t$ where $|V_i| \le k$, for each $1 \le i \le t$. We define super graphs $G_1, G_2, \cdots, G_t$ of $G$ with $G_i(V, E_i)$ defined by setting  $E_i = E \cup \{(x,y)|x,y \in$ $V \setminus V_i\}$, for $1\le i\le t$.
\begin{lemma}\label{th2}
 Let $G_i$ be as defined above, for $1\le i\le t$. Optimal box representation of $G_i$ can be computed in polynomial time. 
\end{lemma}
\begin{proof}
Noting that $G[V\setminus V_i]$ is a clique and $|V_i| \le k=\frac{\sqrt{\log n}}{\sqrt{\log \log n}}$, by Theorem \ref{th4}, we can compute an optimum box representation of $G_i$ in $n^2 2^{{O(k^2 \log k)}}$ $=$ $O(n^3)$ time, where $n=|V|$. 
\qed
\end{proof}
We can compute the optimal box representations of $G_i$, for $1 \le i \le t=  \left \lceil \frac{n\sqrt{\log \log n}}{\sqrt{\log n}} \right \rceil$ as explained in Lemma \ref{th2} in total $O(n^4)$ time. Observe that $E(G)=E(G_1) \cap E(G_2) \cap \cdots \cap E(G_t)$. Therefore, it is a trivial observation that the union of box representations of $G_i$s we computed gives us a valid box representation of $G$. By Lemma \ref{lm4} we have, $box(G_i) \le 2.box(G)$. Hence,
\begin{equation}
box(G) \le box(G_1)+box(G_2)+ \cdots + box(G_t)\le 2t. box(G)
\end{equation}
Substituting $t= \left \lceil \frac{n\sqrt{\log \log n}}{\sqrt{\log n}} \right \rceil$ in the equation above gives the approximation ratio as claimed in Theorem \ref{thgen}.

Using the optimal box representations of $G_i$, for $1 \le i \le t$, a cube representation $\mathcal{C}$ of $G$, such that $|\mathcal{C}| \le t' \cdot cub(G)$, where $t'$ is $O\left(\frac{n {(\log \log n)}^{\frac{3}{2}}}{\sqrt{\log n}}\right)$, can also be computed in polynomial time. We know \cite{Adiga10} that from an optimum box representation of $G_i$, in polynomial time, we can construct a cube representation of $G_i$ of dimension\\$box(G_i) \lceil{\log \alpha(G_i)}\rceil$, where $\alpha(G_i)$ is the independence number of $G_i$ which is at most $|V_i|$. The union of cube representations of $G_i$s we computed gives us a valid cube representation of $G$. Hence,
\begin{eqnarray}
cub(G) &\le& cub(G_1)+cub(G_2)+ \cdots + cub(G_t) \nonumber \\
   &\le& (box(G_1)+box(G_2)+ \cdots + box(G_t)) \lceil \log k \rceil \nonumber \\
&\le& 2t. box(G) O(\log \log n) \le O(t. \log \log n). cub (G) 
 \end{eqnarray}
Substituting $t= \left \lceil \frac{n\sqrt{\log \log n}}{\sqrt{\log n}} \right \rceil$ in the equation above gives the approximation ratio as claimed in Theorem \ref{thgen}.
\subsection{Approximating partial order dimension}
A partially ordered set (poset) $\mathcal{P} = (X, P)$ consists of a nonempty set $X$ and a binary relation $P$ on $X$ that is reflexive,
antisymmetric and transitive. If every pair of distinct elements of $X$ are comparable under the relation $P$, then $(X, P)$ is called a 
total order or a linear order. A linear extension of a partial order $(X, P)$ is a linear order $(X, P')$ such that
$\forall x, y \in X$, $(x, y) \in P \Rightarrow (x, y) \in P'$. The dimension of a poset $\mathcal{P} = (X, P)$, denoted by $dim(\mathcal{P})$ is defined as the 
smallest integer $k$ such that $\mathcal{P}$ can be expressed as the intersection of $k$ linear extensions $(X, P_1), (X, P_2), \ldots, (X, P_k)$ of $\mathcal{P}$: i.e., 
if $\forall x, y \in X$, $(x, y) \in P \Leftrightarrow (x, y) \in P_i$, for each $1 \le i \le k$.
This concept was introduced by Dushnik and Miller in 1941 \cite{Dushnik41}. 

A height-two poset is a poset $(X, P)$ in which all elements of $X$ are either minimal elements or maximal elements under the relation $P$.
Even in the case of height-two posets, partial order dimension is hard to approximate within an $O(n^{1-\epsilon})$ factor for any $\epsilon >0$, unless
$\text{NP}=\text{ZPP}$ \cite{Chalermsook2013}. 
\begin{corollary}
There is a polynomial time algorithm to approximate the partial order dimension of any poset $\mathcal{P} = (X, P)$ defined on a finite set $X$, within an 
$o(n)$ factor, where $n=|X|$.
\end{corollary}
\begin{proof}
Assume that $\mathcal{P} = (X, P)$ defined on a finite set $X$. 

We will first prove the statement for height-two posets.
Adiga et al. \cite{AdigaCOCOON} showed that if $\mathcal{P}$ is a height-two poset defined on a finite set $X$ and $G_P$ is the underlying comparability graph of $\mathcal{P}$ 
(i.e., $X$ is the vertex set of $G_P$ and two vertices are adjacent in $G_P$ if and only if they are comparable under $P$), then 
$\operatorname{box}(G_P) \le dim(\mathcal{P})\le 2 \operatorname{box}(G_P)$. 
Since $\operatorname{box}(G_P)$ can be approximated in polynomial time within an $o(n)$ factor by Theorem \ref{thgen}, 
a polynomial time $o(n)$ factor approximation algorithm for computing the poset dimension of height-two posets follows.

By a construction given by R. Kimble \cite{Trotter78}, given a poset $\mathcal{P} = (X, P)$ of arbitrary height, we can construct a height-two poset 
$\mathcal{P'} = (S(X), P')$ from $\mathcal{P} = (X, P)$ in polynomial time so that $dim(\mathcal{P}) \le dim(\mathcal{P}') \le 1+ dim(\mathcal{P})$. 
Combined with this reduction, the polynomial time $o(n)$ factor approximation algorithm we obtained in the previous paragraph for height-two posets gets 
extended for posets of arbitrary height. 
\end{proof}
\subsection{Approximating the threshold dimension of split graphs}
A graph $G(V, E)$ is called a threshold graph if there exists $s \in \mathbb{R}$ and a labeling of vertices $w: V \mapsto \mathbb{R}$ 
such that $\forall u, v \in V, (u, v) \in E \Leftrightarrow w(u) + w(v) \ge s$. The threshold dimension of $G$, denoted by $t(G)$
is the minimum integer $k$ such that there exists threshold graphs $G_1, G_2, \ldots, G_k$ on the same vertex set as $V(G)$ with 
$E(G) = E(G_1) \cup E(G_2) \cup \cdots \cup E(G_k)$. The concept of threshold graphs and threshold dimension was introduced by 
Chv\'{a}tal and Hammer \cite{Chvatal77} while studying some set-packing problems. 
Threshold dimension is also hard to approximate within an $O(n^{1-\epsilon})$ factor for any $\epsilon >0$, unless
$\text{NP}=\text{ZPP}$ \cite{Chalermsook2013}. The same hardness result holds for the restricted case of split graphs as well \cite{Abh1}. 
\begin{corollary}
There is a polynomial time algorithm to approximate the threshold dimension of any split graph $G$ within an $o(n)$ factor, 
where $n=|V(G)|$.
\end{corollary}
\begin{proof}
 Given any split graph $G$, Adiga et al. \cite{Abh1} gave a polynomial time method to construct another split graph $H$ on the same vertex set
such that $t(G)=\operatorname{box}(H)$. By Theorem \ref{thgen}, the result follows.
\end{proof}
\section{FPT Algorithm for Computing the Cubicity of Graphs with MVC as Parameter} \label{fpt3} 
 In this section, we give an algorithm to compute a cube representation of $G$ which is of dimension at most $2 \cdot cub(G)$, using $MVC(G)$ as parameter $k$, which runs in time $2^{O(2^k k^2)} n^{O(1)}$. In fact, by allowing a larger running time of $2^{O\left(\frac{k^3 2^{\frac{4k}{\epsilon}}}{\epsilon}\right)} n^{O(1)}$, we can achieve a $(1+\epsilon)$ approximation factor, for any $\epsilon >0$.

Let $G(V, E)$ be a graph on $n$ vertices. Without loss of generality, we can assume that $G$ is connected. We can compute \cite{Nie1} a minimum vertex cover of $G$ in time $2^{O(k)}n^{O(1)}$. Let $S \subseteq V$ be a vertex cover of $G$ of cardinality $k$ such that $G(V\setminus S)$ is an independent set on $n-k$ vertices. 
For each $A \subseteq S$, define $N_A =\{v \in (V \setminus S)$ : $N_G(v)=A \}$. Notice that $\{N_A : A \subseteq S$ and $N_A \ne \emptyset\}$ defines a partition of $V \setminus S$.  For each $N_A \ne \emptyset$, let $v_A$ be any arbitrary (representative) vertex in $N_A$. Let $V' = \{ v_A : A \subseteq S$ and $N_A \ne \emptyset \}$. It is easy to see that $|V'| \le 2^k -1$. 

Let $G'$ be the induced subgraph of $G$ on vertex set $V' \cup S$. It is known \cite{ChandranDS09} that $cub(G) \le MVC(G) + \left \lceil \log (|V(G)|-MVC(G))\right \rceil -1$. Since $MVC(G')=k$, we get $cub(G') \le 2k-1$. Using the construction in Lemma \ref{count}, we can compute an optimum cube representation of $G'$ in time $2^{O(2^k k^2)}$. Let $\mathcal{C'} = \{I'_1$, $I'_2$, $\cdots$, $I'_p\}$ be the resultant cube representation of $G'$. Let us construct $p$ unit interval graphs $I_i$, $1 \le i \le p$ as defined below.
$$
\left[ l_v(I_{i}),r_v(I_{i})\right] = 
\begin{cases}  \left[l_{v_A}(I'_i), r_{v_A}(I'_i)\right] & \text{if $v\in N_A$,}
\\
\left[ l_v(I'_{i}), r_v(I'_{i})\right] &\text{if $v \in S$.}
\end{cases}
$$
Since $\mathcal{C'}$ is a cube representation of $G'$, and $\forall v \in N_A$, $N_G(v) = N_G(v_A) = A$, it is easy to verify that $I_i$, $1 \le i \le p$ are unit interval super graphs of $G$.

Let $t= \displaystyle \max_{A \subseteq S} |N_A| $. For each $N_A \ne \emptyset$, let us consider the mapping
$n_A :  N_A  \mapsto \{1, 2, \cdots, |N_A|\}$, where $n_A(v)$ is the unique number representing $v \in N_A$. [Note that if $u \in N_A$ and $v \in N_{A'}$, where $A \ne A'$, then, $n_A(u)$ and $n_{A'}(v)$ could potentially be the same.] For  $1 \le i \le q = \left \lceil \log t \right \rceil$, define 
$b_i : V \setminus S \mapsto \{1, 2, \cdots, q\}$ as $b_i(v) = $ $i^{th}$ bit in the $q$ bit binary representation of $n_A(v)$, when $v \in N_A$.
We define $q$ unit interval graphs $J_1, J_2, \cdot, J_q$ as follows.  
$$\left[ l_v(J_{i}),r_v(J_{i})\right] = 
\begin{cases}  \left[1, 2 \right] & \text{if $v\in S$,}
\\
\left[ 0, 1\right] &\text{if $v \in (V \setminus S)$ and $b_i(v) = 0$,}
\\
\left[ 2, 3\right] &\text{if $v \in (V \setminus S)$ and $b_i(v) = 1$}
\end{cases}
$$
Since in each $J_i$, $1\le i\le q$, $S$ forms a clique in the region $\left[1, 2 \right]$ and intervals corresponding to every $v \in (V \setminus S)$ intersects with this interval, it is easy to see that these are interval super graphs of $G$. These unit interval graphs can be constructed in $O(n \log n)$ time.
\begin{claim}
 $\mathcal{C}=\{I_1, I_2, \cdots, I_p, J_1, J_2, \cdots, J_q \}$ is a valid cube representation of $G$, of dimension $p+q \le 2$ $cub(G)$ constructible in $2^{O(2^k k^2)} n^{O(1)}$ time.
\end{claim}
\begin{proof}
 We already proved that $\mathcal{C}$ is constructible in $2^{O(2^k k^2)} n^{O(1)}$ time and each of the interval graphs in $\mathcal{C}$ is a unit interval super graph of $G$. Now consider any $(u, v) \notin E(G)$. We have the following three cases to analyze.
\begin{itemize}
 \item[(1)] If $u,v \in (V' \cup S)$ with $u \ne v$, $\exists I'_i \in \mathcal{C'}$ such that $(u, v) \notin E(I'_i)$. This implies that, $(u, v) \notin E(I_i)$, because for any $u \in (V' \cup S)$, $\left[ l_v(I_{i}),r_v(I_{i})\right] = \left[ l_v(I'_{i}),r_v(I'_{i})\right]$ by definition.
 \item[(2)] If $u \in (V' \cup S)$ and $v \in V \setminus (V' \cup S)$, then, we know that $ \exists A \subseteq S : v \in N_A$ and $N_G(v) = N_G(v_A) = A$. Therefore, $(u, v_A) \notin E(G)$. Since $\forall 1 \le i \le p$, $\left[ l_v(I_{i}),r_v(I_{i})\right] = \left[ l_{v_A}(I_{i}),r_{v_A}(I_{i})\right]$ with $v_A \in V'$, this case reduces to case 1.
  \item[(3)] If $u, v \in V \setminus (V' \cup S)$ with $u \ne v$, then, $\exists A \subseteq S : u \in N_A$ and $ \exists A' \subseteq S : v \in N_{A'}$. We also know that $\forall 1 \le i \le p$, $\left[ l_u(I_{i}),r_u(I_{i})\right] = \left[ l_{v_A}(I_{i}),r_{v_A}(I_{i})\right]$ and $\left[ l_v(I_{i}),r_v(I_{i})\right] = \left[ l_{v_{A'}}(I_{i}),r_{v_{A'}}(I_{i})\right]$. There are two sub cases to consider.
  \begin{itemize}
   \item[(3a)] If $A \ne A'$, then $(v_A, v_{A'}) \notin E(G)$, with $v_A, v_{A'} \in V'$ and  and hence this case reduces to case 1.
   \item[(3b)] If $A = A'$, then $u, v \in N_A$. We know that the $q$ bit binary representations of $n_A(u)$ and $n_A(v)$ differ at least in one bit position. Therefore, $\exists i \in \{1, 2, \cdots, q\}$ such that $b_i(v) \ne b_i(u)$. Without loss of generality, assume that $b_i(v)=0$ and $b_i(u)=1$. By definition of interval graph $J_i$, $\left[ l_v(J_{i}),r_v(J_{i})\right] = \left[ 0, 1\right]$ and  $\left[ l_u(J_{i}),r_u(J_{i})\right] = \left[ 2, 3\right]$. Thus, $(u,v) \notin E(J_i)$.  
  \end{itemize}
\end{itemize}
It is known \cite{Adiga10} that $cub(G) \ge \lceil \log \psi(G) \rceil$, where $\psi(G)$ is the number of leaf nodes in the largest induced star in $G$. Since $t = \displaystyle \max_{A \subseteq S} |N_A| \le \psi(G)$, we have $q= \lceil \log t \rceil$ $\le \lceil \log \psi(G) \rceil$ $\le cub(G)$. Since $G'$ is an induced subgraph of $G$, we have $ p= cub(G') \le cub(G)$. Thus, $\mathcal{C}$ is a valid cube representation of $G$ of dimension $p+q \le 2$ $cub(G)$, constructible in $2^{O(2^k k^2)} n^{O(1)}$ time.
\qed 
\end{proof}
We can also achieve a $(1+\epsilon)$ approximation factor, for any $\epsilon >0$ by allowing a larger running time as explained below. Define $f(k_\epsilon)=k \left(1+2^{\frac{2k-1}{\epsilon}}\right)$, where $k =MVC(G)$. If $|V(G)|=n \le f(k_\epsilon)$, then, by Lemma \ref{count}, we can get an optimal cube representation of $G$ in time $2^{O(f^2(k_\epsilon) \log f(k_\epsilon))}$. Otherwise, we have $\frac{2k-1} {\left \lceil \log{\left \lceil \frac{n-k}{k} \right \rceil} \right \rceil} \le \epsilon$. In this case, we use the construction described in Section \ref{fpt3}, to get a cube representation of $G$ of dimension $p+q$. We prove that in this case, $p+q \le cub(G) (1+\epsilon)$. 

By pigeon hole principle, $\displaystyle \max_{v \in S} |N_G(v) \cap (V\setminus S) | \ge \left \lceil \frac{n-k}{k} \right \rceil$. Therefore, $cub(G) \ge \lceil \log \psi(G) \rceil \ge \left \lceil \log{\left \lceil \frac{n-k}{k} \right \rceil} \right \rceil$. Recall that $p \le 2k -1$. Therefore, $p+q \le 2k -1 + cub(G)$ $\le cub(G) \left(\frac{2k-1}{cub(G)} + 1\right)$ $\le cub(G) \left(\frac{2k-1} {\left \lceil \log{\left \lceil \frac{n-k}{k} \right \rceil} \right \rceil} +1 \right) \le cub(G) (1+\epsilon)$.\\
The total running time of this algorithm is $2^{O\left(\frac{k^3 2^{\frac{4k}{\epsilon}}}{\epsilon}\right)} n^{O(1)}$.
 
\newpage
\appendix
\section{Appendix}\label{appendix1}
\textbf{Lemma \ref{lmeasy} :} Let $G(V,E)$ be a graph. Let $S\subseteq V$ be such that $\forall v \in V \setminus S$ and $u \in V$, $(u,v) \in E$. If a box representation $\mathcal{B}_S$ of $G[S]$ is known, then, in $O(n^2)$ time we can construct a box representation $\mathcal{B}$ of $G$ of dimension $|\mathcal{B}_S|$. In particular, $box(G) = box(G[S])$.
 \begin{proof}                                                                                                                                                                                                                                                                                                                                                                                                                                                                                                                                                                                                                                                                                                                                                                             
 Let $\mathcal{B}_S = \{I_1$, $I_2$, $\cdots$, $I_p\}$ be a box representation of $G[S]$. For $1\le i \le p$, let $l_i=\displaystyle\min_{u \in S}$ ${l_u(I)}$ and $r_i=\displaystyle\max_{u \in S}$ ${r_u(I)}$. For $1 \le i \le p$ define $I'_i$ by the interval assignment 
$$
\left[ l_v(I'_{i}),r_v(I'_{i})\right] = 
\begin{cases}  \left[ l_v(I_{i}),r_v(I_{i})\right] & \text{if $v\in S$,}
\\
[l_i, r_i] &\text{if $v \in V\setminus S$.}
\end{cases}
$$
It is easy to see that $\mathcal{B}_2 = \{I'_1$, $I'_2$, $\cdots$, $I'_p\}$ is a box representation of $G$ and $box(G) \le box(G[S])$. Since $G[S]$ is an induced subgraph of $G$, we also have $box(G) \ge box(G[S])$. The whole construction can be done in $O(n^2)$ time.
\qed
\end{proof}
\textbf{ Lemma \ref{lm4} :} Let $G(V, E)$ be a graph and let $A \subseteq V$. Let $G_1(V, E_1)$ be a super graph of $G$ with $E_1 = E \cup \{(x,y)|x,y \in A\}$. If a box representation $\mathcal{B}$ of $G$ is known, then in $O(n^2)$ time we can construct a box representation $\mathcal{B}$$_1$of $G_1$ of dimension $2 \cdot |\mathcal{B}|$. In particular, $box(G_1) \le 2.box(G)$.
\begin{proof}
Let $\mathcal{B}$ $=\{I_1$, $I_2$, $\cdots$, $I_b\}$ be a box representation of $G$. For each $1 \le i \le b$, let $l_i = \displaystyle\min_{u\in V}$ $l_u(I_i)$ and $r_i = \displaystyle\max_{u\in V}$ $r_u(I_i)$. For $1\le i \le b$, let $I_{i_1}$ be the interval graph obtained from $I_i$ by assigning the intervals 
$$
\left[ l_v(I_{i_1}),r_v(I_{i_1})\right] = 
\begin{cases}  \left[l_i, r_{v}(I_{i})\right] & \text{if $v\in A$,}
\\
\left[ l_v(I_{i}), r_v(I_{i})\right] &\text{if $v \in V\setminus A$.}
\end{cases}
$$
 and let $I_{i_2}$ be the interval graph obtained from $I_i$ by assigning the intervals
$$
\left[ l_v(I_{i_2}),r_v(I_{i_2})\right] = 
\begin{cases}  \left[ l_{v}(I_{i}),r_i\right] & \text{if $v\in A$,}
\\
\left[ l_v(I_{i}),r_v(I_{i})\right] &\text{if $v \in V\setminus A$.}
\end{cases}
$$
Note that, in constructing $I_{i_1}$ and  $I_{i_2}$ we have only extended some of the intervals of $I_i$ and therefore, $I_{i_1}$ and  $I_{i_2}$ are super graphs of $I$ and in turn of $G$. By construction, $A$ induces cliques in both  $I_{i_1}$ and  $I_{i_2}$, and thus they are supergraphs of $G_1$ too. 

 Now, consider $(u,v) \notin E$ with $u \in V \setminus A$, $v \in A$. Then either $r_{v}(I_i) < l_u(I_i)$ or $r_u(I_i) < l_{v}(I_i)$. If $r_{v}(I_i) < l_u(I_i)$, then clearly the intervals $[l_i, r_{v}(I_i)]$ and $[l_u(I_i), r_u(I_i)]$ do not intersect and thus $(u,v) \notin E(I_{i_1})$. Similarly, if $r_u(I_i) < l_{v}(I_i)$, then $(u,v) \notin E(I_{i_2})$. If both $u, v \in V \setminus A$ and $(u,v) \notin E$, then $\exists i$ such that  $(u,v) \notin E(I_i)$ for some $1\le i\le b$ and clearly by construction,  $(u,v) \notin E(I_{i_1})$ and  $(u,v) \notin E(I_{i_2})$.

  It follows that $G_1=\displaystyle\bigcap_{1 \le i \le b}{I_{i_1} \cap I_{i_2}}$ and therefore, $box(G_1) \le 2 \cdot box(G)$. 
\qed
\end{proof}
\textbf{Lemma \ref{count} :} Let $G(V, E)$ be a graph on $n$ vertices of boxicity (cubicity) $b$. Then an optimum box (cube) representation of $G$ can be computed in $2^{O(nb \log n)}$ time.
\begin{proof} 
In any interval graph on $n$ vertices, the set of end points (both left and right) of intervals corresponding to vertices of $V$ generate an ordering of $2n$ end points. Since this ordering can be done only in at most $2n!= 2^{O(n \log n)}$ ways, we can construct all possible interval graphs on $n$ vertices in $2^{O(n \log n)}$ time. For each of them, in $O(n^2)$ time we can verify that they are interval super graphs of $G$. [In linear time, it is also possible to check whether a given graph is a unit interval graph and if so, generate a unit interval representation of it.] 

All possible box (cube) representations of $G$ of dimension $d$ can be generated in $2^{O(nd \log n)}$ time by choosing any $d$ of the $2^{O(n \log n)}$ (unit) interval super graphs of $G$ at a time. Each box (cube) representation can be validated in $O(dn^2)$ time. We can repeat this for $1 \le d \le b$ in that order to get an optimal box (cube) representation of $G$, where $b$ is the boxicity (cubicity) of $G$. This can be done in $2^{O(nb \log n)}$ time.
\qed
\end{proof} 
\end{document}